\documentclass{article}
\usepackage[utf8]{inputenc}
\usepackage[utf8]{inputenc}
\usepackage{geometry}
\usepackage{babel}
\usepackage{verbatim}
\usepackage{mathtools}
\usepackage{bm}
\usepackage{amsmath}
\usepackage{amssymb}
\usepackage{multicol}
\usepackage[ruled,noline,noend,linesnumbered]{algorithm2e}
\usepackage[dvipsnames,table,xcdraw]{xcolor}
\definecolor{linkColor}{HTML}{E74C3C}
\definecolor{pearcomp}{HTML}{B97E29}
\definecolor{citeColor}{HTML}{2980B9}
\definecolor{urlColor}{HTML}{1D2DEC}
\definecolor{conjColor}{HTML}{9ab569}
\usepackage[CJKbookmarks=true,
            bookmarksnumbered=true,
            bookmarksopen=true,
            colorlinks=true,
            citecolor=citeColor,
            linkcolor=linkColor,
            anchorcolor=red,
            urlcolor=urlColor,
            ]{hyperref}
\makeatletter
\usepackage{bm}
\usepackage{enumitem}
\usepackage{tikz}
\usetikzlibrary{calc,patterns,angles,quotes}

\tikzset{
    right angle quadrant/.code={
        \pgfmathsetmacro\quadranta{{1,1,-1,-1}[#1-1]}     % Arrays for selecting quadrant
        \pgfmathsetmacro\quadrantb{{1,-1,-1,1}[#1-1]}},
    right angle quadrant=1, % Make sure it is set, even if not called explicitly
    right angle length/.code={\def\rightanglelength{#1}},   % Length of symbol
    right angle length=2ex, % Make sure it is set...
    right angle symbol/.style n args={3}{
        insert path={
            let \p0 = ($(#1)!(#3)!(#2)$) in     % Intersection
                let \p1 = ($(\p0)!\quadranta*\rightanglelength!(#3)$), % Point on base line
                \p2 = ($(\p0)!\quadrantb*\rightanglelength!(#2)$) in % Point on perpendicular line
                let \p3 = ($(\p1)+(\p2)-(\p0)$) in  % Corner point of symbol
            (\p1) -- (\p3) -- (\p2)
        }
    }
}

\usepackage{babel}
\usepackage{subcaption}

\usepackage[authoryear,sort]{natbib}
\usepackage{algorithmic}

\usepackage{amsthm}
\usepackage{bbm}
\usepackage{times}
\usepackage{lipsum}

\theoremstyle{plain}

\newtheorem{lemma}{\textbf{Lemma}}
\newtheorem{theorem}{\textbf{Theorem}}

\newtheorem{assumption}{\textbf{Assumption}}

\newtheorem{proposition}{\textbf{Proposition}}
\newtheorem{procedure_con}{\textbf{Procedure}}
\theoremstyle{definition}
\newtheorem{remark}{\textbf{Remark}}

\DeclareMathOperator*{\argmax}{argmax}

\usepackage[authoryear,sort]{natbib} 

\title{Online Learning in a Creator Economy}

\author{Banghua Zhu, Sai Praneeth Karimireddy, Jiantao Jiao, Michael I. Jordan\thanks{Banghua Zhu, Sai Praneeth Karimireddy, Jiantao Jiao, Michael I. Jordan are with the Department of Electrical Engineering and Computer Sciences, University of California, Berkeley. Email: \{banghua, sp.karimireddy, jiantao, jordan\}@berkeley.edu.}}
\date{}

\begin{document}

\maketitle

\begin{abstract}
The creator economy has revolutionized the way individuals can profit through online platforms.   In this paper, we initiate the study of online learning in the creator economy by modeling the creator economy as a three-party game between the users, platform, and content creators, with the platform interacting with the content creator under a principal-agent model through contracts to encourage better content. Additionally, the platform interacts with the users to recommend new content, receive an evaluation, and ultimately profit from the content, which can be modeled as a recommender system.

Our study aims to explore how the platform can jointly optimize the contract and recommender system to maximize the utility in an online learning fashion. We primarily analyze and compare two families of contracts: return-based contracts and feature-based contracts. Return-based contracts pay the content creator a fraction of the reward the platform gains. In contrast, feature-based contracts pay the content creator based on the quality or features of the content, regardless of the reward the platform receives. We show that under smoothness assumptions, the joint optimization of return-based contracts and recommendation policy provides a regret $\Theta(T^{2/3})$. For the feature-based contract, we introduce a definition of intrinsic dimension $d$ to characterize the hardness of learning the contract and provide  an upper bound on the regret $\mathcal{O}(T^{(d+1)/(d+2)})$. The upper bound is tight for the linear family.

\end{abstract}

\section{Introduction}

The \emph{creator economy} refers to a rapidly growing online-platform-facilitated economy that brings together content creators and users, allowing creators  to earn revenue from their creations~\citep{banks2008labour, schram2020state, el2022quantifying, radionova2021creator, bhargava2022creator}. These platforms monetize the content created by content creators through various means, including paid audience partnerships, ad revenue, tipping platforms, and product sales provided by the users.

The creator economy can be viewed as  a  three-party game linking users, platform, and content creators. On the one hand, we can model the interactions between the platform and the content creator as a principal-agent relationship, focusing on the need to incentivize the production of high-quality content by the content creator. The platform 
would like to collect better content to enhance the desirability of the platform. The content creator wishes to gain profit on the platform from their content. The two sides develop agreements in the form of a \emph{contract}, which specifies how much the platform would pay under the different possible kinds of content. By learning the intent and interest of the content creator, the platform is able to identify a better way to incentivize participation and share the profit with the content creator.  The overall framework is \emph{contract theory}, which is a branch of the theory of incentives~\citep{grossman1992analysis, faure2001transaction, bolton2004contract, salanie2005economics}.

On the other hand, with the created content, the platform interacts with the users to recommend new content, and receives evaluation and revenue from the content, an interaction which can be modeled as a recommender system~\citep{resnick1997recommender, lu2015recommender, lu2012recommender, zhang2019deep, burke2002hybrid, jannach2010recommender, ricci2010introduction, schafer1999recommender}. The platform aims at designing the recommendation algorithms that are personalized for each type of user based on their interests and historical behaviors on the platform.  The goal is to maximize user satisfaction with the recommended content, which leads to an increase in total social welfare.

While much attention has been paid to the platform's interactions with the user, there has been very little focus on the design of contracts between the platform and the content creator. 
In this paper, we aim to fill this gap by initiating the study of  the joint design of contract and recommendation algorithms. We formulate the problem as a repeated game between the platform, users, and content creators. In each round, the platform proposes a new contract to all content creators. Based on the contract, the content creators are incentivized to generate content. The platform then recommends part of the content to each user and receives a corresponding reward. We aim at minimizing an appropriate notion of regret, in particular the difference in the utility between the chosen contract and recommendation policy with  the optimal contract and recommendation policy. 
  We establish a theory that helps understand the dynamics of the creator economy and sheds light on how the platform  can effectively incentivize content creators to produce better content while maximizing their own profits. By examining different types of contracts and their impact on the utility, our study provides valuable insights into the workings of a creator economy.

\subsection{Main Results}

We analyze regret bounds for the joint design of contract and recommendation algorithms. We focus on the tabular case where one estimates the reward for each user and the created content from the content creator separately. We analyze the two families of contracts:
\begin{itemize}
    \item \textbf{Return-based contract.} The payment to the content creator is proportional to the reward the platform receives from the user. 
    \item \textbf{Feature-based contract.} The payment to the content creator is based on the features capturing the quality of the content, regardless of the reward the platform receives.
\end{itemize}
We are interested in the regret for both contracts when combined with the design of the recommendation algorithm. Our main theoretical results are summarized below.
\begin{theorem}[Informal]
   Let $T$ be the number of total rounds. Under certain smoothness assumptions, the regret for joint optimization of return-based contract and recommendation policy is $\Theta(T^{2/3})$.
   
   On the other hand, the regret for joint optimization of a feature-based contract and recommendation policy is $\mathcal{O}(T^{(d+1)/(d+2)})$, where $d$ is the intrinsic dimension that depends on the complexity of the contract family.  The upper bound is tight when the contract family is linear.
\end{theorem}

Based on our theoretical analysis, we are able to draw the following conclusions:
\begin{itemize}
\item Both the recommendation policy and the contract design play an important role in the utility of the platform. Learning and adapting to the behavior of content creators is important to maximize the overall utility of the platform.
\item The optimal feature-based contract requires more samples to learn     compared to the return-based contract. However, when the feature-based contract family is well-designed and the content creators always maximize their own utility, the feature-based contract can achieve a larger optimal utility in the full recommendation setting than that of the return-based contract.  
\end{itemize}

\subsection{Related Work}
\paragraph{Contract Theory.}
Contract theory has been a subject of study for several decades~\cite{grossman1992analysis, faure2001transaction, bolton2004contract, salanie2005economics}, with a vast body of literature exploring its connections to information design~\citep[e.g.,][]{kamenica2011bayesian} and mechanism design~\citep[e.g.,][]{hurwicz2006designing}. The most relevant study to our paper is the online learning aspect of contract design. \citet{ho2016adaptive} initiated the study of learning the optimal contract design in an online fashion by formulating the problem as a repeated principal-agent model and applying an adaptive zooming algorithm in a bandit setting to obtain regret bounds. They provided upper bounds on the regret under the assumption of monotone contract and  first-order stochastic dominance among the probabilities. \citet{cohen2022learning} extended the analysis to  the case when the utility of the agent has bounded risk-aversion. \citet{zhu2022sample} pinned down the sample complexity (and regret) for linear contracts and provided a nearly tight bound for general contracts without extra assumptions.

Our definitions of return-based contract and feature-based contract are closely related to the linear contract and general contract studied in the above literature~\citep{ho2016adaptive, zhu2022sample}. The return-based contract is exactly the linear contract, which pays the agent a fraction of the reward received by the principal. Our feature-based contract can be viewed as a generalization of the general contract, where the latter specifies the payment for all possible discrete outcomes, and the former is an arbitrary function of the (possibly continuous) outcomes.

Beyond the learning perspective, 
there has been interest in exploring linkages between statistical inference and contract theory~\citep{schorfheide2012use, schorfheide2016hold, frazier2014incentivizing, bates2022principal}. The computation and approximation aspects have also been actively explored in recent years~\citep{babaioff2006combinatorial, dutting2021complexity, guruganesh2021contracts, castiglioni2022designing, dutting2019simple, castiglioni2021bayesian, carroll20152014}.   

\paragraph{Recommender systems and content creators in a creator economy.}
 There is a large literature on recommender systems~\citep{resnick1997recommender, lu2015recommender, lu2012recommender, zhang2019deep, burke2002hybrid, jannach2010recommender, ricci2010introduction, schafer1999recommender}, with  two main streams of system design---collaborative filtering~\citep{breese2013empirical} and content-based filtering~\citep{aggarwal2016recommender}. In collaborative filtering, the system recommends the same item to users who are linked in some manner. In content-based filtering, the system recommends items similar to those that a user liked in the past. In our paper, we consider the setting where the users are pre-grouped into different types, and assume that users with the same type share the same preferences, similar to the spirit of collaborative filtering. We also note that there is a recent line of work that studies the strategic behavior of content creators in a competitive model and characterizing equilibria~\citep{qian2022digital, jagadeesan2022supply, yao2023bad}.

\paragraph{Continuum-armed bandit and performative prediction.}
Our work is closely related to the continuum-armed bandit problem~\citep{kleinberg2008multi, kleinberg2013bandits, lattimore2020bandit} and performative prediction problem~\citet{hardt2016strategic, zrnic2021leads, liu2016bandit}. In the continuum-armed bandit, one first discretizes the choices of actions and reduces the problem to a finite-armed bandit problem. For the one-dimensional continuum-armed bandit problem, the optimal regret is $\Theta(T^{2/3})$, which can be achieved by an upper confidence bound algorithm after discretization. In the current paper we show that the extra learning of optimal recommendation policy does not affect the exponent in $T$ in this paper, and generalize this insight to the case of feature-based contract, where the content generation function is assumed to lie in a Lipschitz function family.

The literature of strategic classification and performative prediction also includes smoothness assumption in the utility function~\citep{hardt2016strategic, dong2018strategic,chen2019grinding, zrnic2021leads, liu2016bandit, harris2022strategic, yu2022strategic, perdomo2020performative}.
\citet{jagadeesan2022regret} study the problem of regret minimization in performative prediction. Our formulation goes beyond this work in that it requires the joint optimization of both the contract design and recommendation algorithm.

\section{Problem Formulation}

\subsection{Notation}
We use $\mathbb{B}^d\subset [0,1]^d$ to denote the $d$-dimensional unit ball.
We use $[K]$ to denote the set $\{1,\cdots, K\}$. We let $\mathsf{Bern}(p)$ be the Bernoulli distribution  with parameter $p$.  We use $\mathsf{TV}$ to denote the total variation distance between two distributions,  and let ${\mathsf{KL}}(\mathbb{P}_1, \mathbb{P}_2) = \int_{\mathcal{X}}  \log(d\mathbb{P}_1/d\mathbb{P}_2)d \mathbb{P}_1$ be the $\mathsf{KL}$ divergence between $\mathbb{P}_1$ and $\mathbb{P}_2$ when $\mathbb{P}_1$ is absolutely continuous with respect to $\mathbb{P}_2$.  We write $f(x) = \mathcal{O}(g(x))$ or $f(x)\lesssim g(x)$ for $x\in A$  if   there exists some positive real number $C$   such that $f(x)\leq C g(x)$ for all $x
\in A$. We write $f(x) = \Omega(g(x))$ or $f(x)\gtrsim g(x)$ for $x\in A$ if   there exists some positive real number $C$   such that $f(x)\geq C g(x)$ for all $x\in A$. We write $f(x) = \Theta(g(x))$ or $f(x)\asymp g(x)$ for $x\in A$  if we have both $f(x) = \mathcal{O}(g(x))$  and $f(x) = \Omega(g(x))$ for $x\in A$. We  use $\widetilde{\Theta}$ to denote the same order up to logarithmic factors. 
\subsection{Single-round interaction}
We begin with a single-round interaction procedure.  The platform (principal) first announces a contract  that specifies the payment to the content creators. Based on the contract, the content creators are incentivized to create content of various qualities. After seeing the contents, the platform will recommend a fraction of them to each user, and gain profits from them. 
The platform aims at designing a good contract and recommendation policy  to encourage better content and to improve its own utility. We formally define these concepts as follows.
\begin{itemize}
    \item {\it Content:}
We assume that one item of content is generated from each of the $K$ content creators. Each item is represented by a vector $c\in\mathcal{C}\subset \mathbb{B}^d$ representing different dimensions of qualities and features of the content.

\item {\it User}: Assume that there are in total $M$ users. Each user is represented by a type $j\in [J]$. In each round, each user will be recommended $S$ items with $S\leq K$.  For each of the recommended items $c\in\mathcal{C}$, the user with type $j$ will provide a random variable $R(j, c)$, in the form of ratings or tips, as the  reward for   the platform. We assume that the reward is always bounded between $0$ and $1$, i.e. $R(j, c)\in[0,  1]$ for any $j\in[J], c\in\mathcal{C}$.  The total reward for the $k$-th content creator can be informally written as $\sum_{j=1}^M R(j, c_k)\cdot \mathbbm{1}(c_k \text{ is recommended to user $j$})$.

\item {\it Recommendation}: Given a set of items $\{c_{k}\}_{k=1}^K$ and a set of users $\{j\}_{j=1}^M$, we let  the  recommendation policy $\pi:[M]\times [K]\times  \mathcal{C}^K\mapsto \{0, 1\}$ be the indicator of whether the content is recommended to the user. A valid recommender policy must satisfy $\pi\in\Pi=\{\pi \mid \sum_{k=1}^K \pi(j, k, \{c_{k'}\}_{k'=1}^K)=S, \forall k \in[K], j\in[M]\}$. The total reward for the $k$-th content creator can  be formally written as $ \sum_{j=1}^M R(j, c_k)\cdot\pi(j, k, \{c_{k'}\}_{k'=1}^K)$. The platform observes each reward $R(j, c_k)$
if the $k$-th content is recommended to the $j$-th user ($\pi(j, k, \{c_{k'}\}_{k'=1}^K) = 1$). When the family  of contents   $\{c_{k'}\}_{k'=1}^K $ is clear from the context, we abbreviate the recommendation policy as $\pi(j, k)$. 

\item {\it Contract}: The  platform proposes a contract $f:\mathcal{C}\times\mathbb{R}\mapsto [0, 1]$  that specifies the payment for the content and reward. For any content $c\in\mathcal{C}$, any reward $r\in\mathbb{R}$ gained by the content, $f(c, r)$ is the payment to the  corresponding content creator based on the content $c$  and the reward $r$ she has produced.  

\item {\it Creation of Content: } The content created is affected by the choice of contract provided by the platform. We  assume that the content generated from the $k$-th content creator is a deterministic function of the provided contract $c_k(f)$. We abbreviate the content as $c_k$ if the contract is clear from  the context.

\end{itemize}

The principal would like to design the contract $f$ and the recommender algorithm $\pi$ such that its own expected utility is maximized. The utility of the principal is defined as
\begin{align}
     u(f, \pi) & = \mathbb{E}\left[\sum_{k=1}^K  R_{\mathsf{total}}(c_k(f), \pi) - f(c_k(f), R_{\mathsf{total}}(c_k(f), \pi))\right],\\
     \text{where } R_{\mathsf{total}}(c_k(f), \pi) & = \sum_{j=1}^M R(j, c_k(f))\cdot\pi(j, k, \{c_{k'}(f)\}_{k'=1}^K).
\end{align}
The optimal contract and recommendation policy can  be found by
\begin{align*}
    (f^\star, \pi^\star) = \argmax_{f\in\mathcal{F}, \pi\in\Pi} u(f,\pi).
\end{align*}

Consider two families of contracts:
\begin{itemize}

    \item Return-based  contract: $\mathcal{F} = \{f(c, r) = \alpha \cdot r\mid \alpha \in[0,1]\}$, which ignores the identity of the content creator and the quality of the content, and directly pays the content creator an $\alpha$ fraction of the  return gained from the content.
    \item Feature-based contract: $\mathcal{F} = \{f(c, r) = g(c) \mid g\in\mathcal{G}\}$, which ignores the total reward and pays the content creator according to the featural representation of the content.
\end{itemize}

In the case of the return-based contract, the utility can be written as
\begin{align*}
    u_{\mathsf{r}}(\alpha, \pi) &= (1-\alpha) \mathbb{E}\left[\sum_{k=1}^K  \sum_{j=1}^M R(j, c_k(\alpha))\cdot\pi(j, k, \{c_{k'}(\alpha)\}_{k'=1}^K)\right]. 
\end{align*}

In the case of the feature-based contract, the utility can be written as
% \begin{align*}
% u_{\mathsf{f}}(\theta, \pi) & =\sum_{k=1}^K r(c^{\star}_k, \pi) - f(c^{\star}_k, r(c^{\star}_k, \pi)) \\
%     & = \sum_{k=1}^K \sum_{j=1}^M u_j^\top c^{\star}_k\cdot\pi(j, c^{\star}_k)-\theta^\top c^{\star}_k, \\
%     \text{ where } c^{\star}_k& = \argmax_{c\in\mathcal{C}^i} \theta^\top c^{\star}_k - l_k(c). 
% \end{align*}
% One may also consider a customized feature-based contract, which proposes a different contract for each content creator. The  utility is 
\begin{align*}
u_{\mathsf{f}}(g, \pi) &  = \mathbb{E}\left[\sum_{k=1}^K \left(\left(\sum_{j=1}^M R(j, c_k(g))\cdot\pi(j, k, \{c_{k'}(g)\}_{k'=1}^K)\right)-g(c_k(g))\right)\right]. 
\end{align*}

Our definition of the content creation process is different from the traditional definition in contract theory~\citep{grossman1992analysis, faure2001transaction, bolton2004contract, salanie2005economics}. In traditional contract theory, the generated content $c_k$ is the  content (or distribution of content) that maximizes the utility of the content creator, which is defined as the expected gain from the content subtracting the cost of producing the content. However, with the recommendation policy, the utility of the content creator depends not only on her own product but also on the content created by other creators. This introduces an extra layer of complication when defining the responses. Thus we  assume that the created content is a deterministic function of the proposed contract, and leave the randomness in the observed reward. We provide more discussion in Section~\ref{sec:full}. % It is an interesting open problem how one can generalize the analysis to the case with randomly created content. 

\subsection{Multi-round interaction}
Consider a repeated interaction procedure where the content creator produces new content, and the  platform recommends the content to the user in each round. % Since there is only one content creator and one user, the platform always recommends the content to the user. The problem reduces to the contract design problem. 
The platform aims at designing a contract with the content creator in each round to encourage better content. The whole procedure  is summarized in Procedure~\ref{procedure.contract}. 
\begin{procedure_con}\label{procedure.contract}
Procedure at round $t$:
\begin{enumerate}
    \item The platform announces a  contract $f^{(t)}$ from the set $\mathcal{F}$ based on the prior information 
    \begin{align*}
    \mathcal{H}^{(t-1)}& =\Big(f^{(1)},  \{c_{k}^{(1)}\}_{k=1}^K, \pi^{(1)}, \{R^{(1)}(j, c_k)\cdot\pi^{(1)}(j, k, \{c_{k'}\}_{k'=1}^K)\}_{j\in[J], k\in[K]}, \cdots, \\ 
   &  \quad f^{(t-1)},  \{c_{k}^{(t-1)}\}_{k=1}^K, \pi^{(t-1)},   \{R^{(t-1)}(j, c_k)\cdot\pi^{(t-1)}(j, k, \{c_{k'}\}_{k'=1}^K)\}_{j\in[J], k\in[K]}\Big).
    \end{align*}
    \item Each content creator generates a content item   $ c_k^{(t)} = c_k(f^{(t)}).$
    \item After observing $\{  c_k^{(t)}\}_{k\in[K]}$, the platform determines a recommendation policy $\pi^{(t)}\in \Pi$ and recommends content  $c_k^{(t)}$ to user $j$ when $\pi^{(t)}(j, c_k^{(t)}) = 1$. 
\item The platform observes the individual reward for recommended content $R(j, c_k^{(t)})$ if $\pi^{(t)}(j, c_k^{(t)}, \{c_{k'}^{(t)}\}_{k'\in[K]}) = 1$, and pays the $k$-th content creator  $f^{(t)}(c_k^{(t)}, R_{\mathsf{total}}(c_k^{(t)}, \pi^{(t)}))$.
\end{enumerate}
\end{procedure_con}

We are interested in minimizing the following regret of utility for the platform, based on either the return-based contract or feature-based contract:
\begin{align*}
    \mathsf{Regret}_{\mathsf{r}}(T) & =\sum_{t=1}^T u_{\mathsf{r}}(\alpha^\star, \pi^\star) - u_{\mathsf{r}}(\alpha_t, \pi_t), \\ 
    \mathsf{Regret}_{\mathsf{f}}(T) & = \sum_{t=1}^T u_{\mathsf{f}}(g^\star, \pi^\star) - u_{\mathsf{f}}(g_{t}, \pi_t).
\end{align*}

\section{Joint Design of Contract and Recommendation Policy}
We turn to an analysis of the regret for both return-based and feature-based contracts. 
\subsection{Return-based contract}
For a return-based contract, the contract can be represented as a scalar $\alpha\in[0, 1]$. 
We make the following assumption on the smoothness of the reward function with respect to the contract. 
\begin{assumption}\label{asm:alpha_lipschitz}
Assume that the reward and  content generation functions satisfy the following smoothness condition, for any $k\in[K]$ and any $\alpha, \alpha'\in[0, 1]$:
\begin{align*}
     \mathbb{E}[R(j, c_k(\alpha))] - \mathbb{E}[R(j, c_k(\alpha'))] \leq L|\alpha - \alpha'|.
\end{align*}
\end{assumption}
As a sufficient condition, 
such assumption holds when both $c_k$ and $R(j, \cdot)$ are Lipschitz.   
With the Lipschitzness of the expected reward,  
we show that an Upper-Confidence-Bound (UCB) based  algorithm achieves nearly optimal regret for a return-based contract.

\begin{algorithm}[!htbp]	\caption{Online Learning with Return-based Contract}
	\label{alg:return}
	\textbf{Input:} The total rounds $T$.\\ 
	Set $\epsilon = T^{-1/3}$.  Define the uniformly discretized set of parameter $\mathcal{S}_\epsilon = \{0, \epsilon, 2\epsilon, \cdots, 1-\epsilon \}$. Set $\hat r^{(0)}(j, c_k(\alpha)) =1$ for all $j\in[M]$,   $k\in[K]$, $\alpha\in\mathcal{S}_\epsilon$. \\
\textbf{For} $t\in\{0, 1,2,\cdots, T\}$: \\
\quad \textbf{If } $t< 1/\epsilon$: \\
\qquad Select the contract $\alpha_t = \epsilon t$, observe  $c_k(\alpha_t)$ for all $k\in[K]$, take arbitrary recommendation policy.\\
\quad \textbf{Else}: \\
\qquad Set 
\begin{align*}
    \hat u_{\mathsf{r}}^{(t)}(\alpha, \pi) &=  (1-\alpha)  \sum_{k=1}^K  \sum_{j=1}^M \hat r^{(t)}(j, c_k(\alpha))\cdot\pi(j, k, \{c_{k'}(\alpha)\}_{k'=1}^K).
\end{align*} \\
 \qquad Select the contract  and the recommendation policy as  $$\alpha_t, \pi_t = \mathsf{argmax}_{\alpha\in\mathcal{S}_\epsilon, \pi\in \Pi} \hat u_\mathsf{r}^{(t)}(\alpha, \pi).$$

\quad Observe the generated content $c_k(\alpha_t)$ and the reward  $R^{(t)}(j, c_k(\alpha_t))$  for the recommended items with $\pi_t(j, k)=1$. Update the estimate $\hat r^{(t+1)}(j, c_k(\alpha))$  with the UCB-based estimator
\begin{align*}
\hat   r^{(t+1)}(j, c_k(\alpha)) &= \begin{cases}
1, \qquad \qquad  \qquad  \qquad  \qquad  \qquad  \qquad    \text{if }  \sum_{s=1}^{t} \mathbbm{1}(\pi_s(j, k)=1,  \alpha_s = \alpha) = 0, \\
\frac{\sum_{s=1}^{t} \mathbbm{1}(\pi_s(j, k)=1,  \alpha_s = \alpha) \cdot R^{(s)}(j, c_k(\alpha_s)) }{ \sum_{s=1}^{t} \mathbbm{1}(\pi_s(j, k)=1,  \alpha_s = \alpha)} + \sqrt{\frac{2\log(MKT/\epsilon\delta)}{\sum_{s=1}^{t} \mathbbm{1}(\pi_s(j, k)=1,  \alpha_s = \alpha)}}, 
 \text{otherwise}.
\end{cases}
\end{align*}
\vspace{12pt}
\end{algorithm}

\begin{theorem}\label{thm:return}
Under Assumption~\ref{asm:alpha_lipschitz}, 
the regret for Algorithm~\ref{alg:return} with a return-based contract satisfies the  following with probability at least $1-\delta$:
   \begin{align*}
    \mathsf{Regret}_{\mathsf{r}}(T)  &\leq  C( KM(L+1) + \sqrt{ { 2KMS\log(KMT/\delta)} })\cdot T^{2/3}.
    \end{align*}
Furthermore, the expected regret satisfies
    \begin{align*}
    \mathbb{E}[\mathsf{Regret}_{\mathsf{r}}(T)] & = \widetilde \Theta(T^{2/3}).
    \end{align*}
\end{theorem}

The proof is deferred to Appendix~\ref{proof:return}. A few remarks are provided below. 
\begin{remark}
Algorithm~\ref{alg:return} starts with a uniform discretization of the contract space.  For each contract in the discretized space, we maintain an estimate of the reward for recommending the contract to each user. Our selection of the contract and recommendation policy is based on a  top-$S$ selection of their upper confidence bound.  We show that this extra layer of learning a recommendation policy does not hurt the dependency on $T$ compared to the traditional one-dimensional continuum-armed bandit problem~\citep{kleinberg2003value, kleinberg2013bandits}. 
\end{remark}
\begin{remark}
Our algorithm and analysis focus on   random reward and deterministic content generation functions. 
This enables fast learning of the content generation function and the main difficulty in learning is the estimation of the expected reward. A natural extension is to deal with both random reward and random content generation functions. The main challenge here is that  the expected reward is more complicated due to the maximum operation in recommendation algorithms, which first observe the samples from the content generation process and then make decisions.   
\end{remark}

\subsection{Feature-based contracts}

For feature-based contracts, we assume that the contract lies in a family $\mathcal{G}$. 
We make the following assumption on the smoothness of the related functions. 
\begin{assumption}[Smoothness of reward, content and contract]\label{asm:lipschitz_g}
Assume that the reward, content generation functions and the contract functions are all Lipschitz; i.e.,  for any $j\in[M], k\in[K]$ and any $g, g'\in\mathcal{G}$, and $c, c'\in \mathcal{C}$,
\begin{align*}
     \mathbb{E}[R(j, c)] - \mathbb{E}[R(j, c')] & \leq L_1\|c-c'\|_2, \\
     |g(c) - g(c')| & \leq L_2 \|c-c'\|_2,\\
     \|c_k(g) - c_k(g')\|_2 & \leq  L_3\|g - g'\|_\infty.
\end{align*}
\end{assumption}

Furthermore, we make an assumption on the complexity of the set $\mathcal{G}$.

\begin{assumption} [Intrinsic dimension $d(\mathcal{G})$]\label{asm:dim_g}
Let $\mathcal{G}_\epsilon$ be the $\epsilon$-covering of $\mathcal{G}$ under infinity norm, i.e. for any $g\in\mathcal{G}$, we can always find some $g'\in\mathcal{G}_\epsilon$ such that have $ \|g-g'\|_\infty \leq \epsilon$. 
 We define the intrinsic dimension of the contract design problem using the metric entropy of the space $\mathcal{G}$:
 \begin{align*}
     d(\mathcal{G}) \coloneqq  \sup_{\epsilon<0.1}\frac{\log |\mathcal{G}_\epsilon|}{\log(1/\epsilon)}.
 \end{align*}
\end{assumption}
 
Given these two assumptions,  
we show that a  UCB-type algorithm achieves  bounded regret for feature-based contracts. 
% \begin{theorem}
% For the case of $L<K$, even when the principal is aware of the users' feature $\{u_j\}_{j=1}^M$ and $\hat r_k(c) =\sum_{j=1}^M u_j^\top c$, the regret for learning a linear contract satisfies
% \begin{align*}
%       \mathsf{Regret}_{\mathsf{f}}(T) = \Omega(T^{(d+1)/(d+2)}).
% \end{align*}
% \end{theorem}
% \begin{proof}

% \end{proof}
% On the other hand, we show that a noisy behavior of the content creator helps smooth the utility, and thus enables efficient learning. Consider a softmax best response instead of a traditional best response. In the linear case, we assume that the created content is a sample from a distribution, which satisfies
% \begin{align*}
%     \mathbb{P}(c^\star_k = c) = \frac{\exp((\alpha\cdot \hat r_k(c)- l_k(c))/\sigma)}{\sum_{c'\in\mathcal{C}}\exp((\alpha\cdot \hat r_k(c')- l_k(c'))/\sigma)}. 
% \end{align*}
% Similarly, for the feature-based contract, we have
% \begin{align*}
%     \mathbb{P}(c^\star_k = c) = \frac{\exp((\theta_k^\top c - l_k(c))/\sigma)}{\sum_{c'\in\mathcal{C}}\exp((\theta_k^\top c' - l_k(c'))/\sigma)}. 
% \end{align*}

% Under such noisy behavior, we show that the regret can be improved.

\begin{algorithm}[!htbp]	\caption{Online Learning with Feature-based Contract}
	\label{alg:linear}
	\textbf{Input:} The total rounds $T$.  \\ 
	Set $\epsilon = T^{-1/(d(\mathcal{G})+2)}$.  Take $\mathcal{G}_\epsilon$ as the $\epsilon$-covering of $\mathcal{G}$ under the infinity norm.  Set $\hat r^{(0)}(j, c_k(g)) =1$ for all $j\in[M]$,   $k\in[K]$, $g\in\mathcal{G}_\epsilon$. \\
\textbf{For} $t\in\{0, 1,2,\cdots, T\}$: \\
\quad \textbf{If } $t< 1/\epsilon$:  \\
\qquad Select the contract $g_t$ as the $t$-th item in $\mathcal{G}_\epsilon$, observe  $c_k(g_t)$ for all $k\in[K]$, take arbitrary recommendation policy.\\
\quad \textbf{Else}: \\
\qquad Set 
\begin{align*}
    \hat u_{\mathsf{f}}^{(t)}(g, \pi) &=    \sum_{k=1}^K \left(  \sum_{j=1}^M \hat r(j, c_k(g))\cdot\pi(j, k, \{c_{k'}(g)\}_{k'=1}^K) - g(c_{k}(g))\right).
\end{align*} \\
 \qquad Select the contract  and the recommendation policy as  $$g_t, \pi_t = \mathsf{argmax}_{g\in\mathcal{G}_\epsilon, \pi\in \Pi} \hat u_\mathsf{r}^{(t)}(g, \pi).$$

\quad Observe the generated content $c_k(g_t)$ and the reward  $R^{(t)}(j, c_k(g_t))$  for the recommended items with $\pi_t(j, k)=1$. Update the estimate $\hat r^{(t+1)}(j, c_k(g))$  with UCB-based estimator
\begin{align*}
\hat   r(j, c_k(g)) &= \begin{cases}
0,   \qquad \qquad \qquad \qquad \qquad \qquad \quad \text{if }  \sum_{s=1}^{t} \mathbbm{1}(\pi_s(j, k)=1,  g_s = g) = 0, \\
\frac{\sum_{s=1}^{t} \mathbbm{1}(\pi_s(j, k)=1,  g_s = g) \cdot R^{(s)}(j, c_k(g_s) }{ \sum_{t=1}^{t} \mathbbm{1}(\pi_s(j, k)=1,  g_s = g)} + \sqrt{\frac{2\log(MKT/\epsilon\delta)}{\sum_{s=1}^{t} \mathbbm{1}(\pi_s(j, k)=1,  g_s = g)}},   \text{otherwise }.
\end{cases}   
\end{align*}
\vspace{12pt}
\end{algorithm}

\begin{theorem}\label{thm:feature}
Under Assumption~\ref{asm:lipschitz_g} and~\ref{asm:dim_g}, the regret for feature-based contract satisfies that with probability at least $1-\delta$,
    \begin{align*}
\mathsf{Regret}_{\mathsf{f}}(T) & \leq  C\cdot (KML_1L_3 + KL_2L_3 +  \sqrt{ KMd(\mathcal{G})\log(KMT/\delta)})\cdot T^{(d(\mathcal{G})+1)/(d(\mathcal{G})+2)}.
    \end{align*}
\end{theorem}
The proof is deferred to Appendix~\ref{proof:feature}. A few remarks  are provided below.
\begin{remark}
    As a concrete example, consider the family of contracts and rewards that are linear with the generated content. We have $ \mathcal{C} = \mathbb{B}^d$,  $\mathcal{G} = \{g_\theta(c) = \theta^\top c, \forall c\in\mathcal{C} \mid \theta\in\mathbb{B}^d\}$, $\mathbb{E}[R(j, c)] \in \{r(c) = \theta^\top c, \forall c\in\mathcal{C} \mid \theta\in\mathbb{B}^d \}$, $c_k(g) = \argmax_{c\in\mathcal{C}} g(c)$.  This family satisfies both Assumption~\ref{asm:lipschitz_g} and~\ref{asm:dim_g} with $L_1= L_2=L_3 =1$, and $d(\mathcal{G})=d$. Thus Theorem~\ref{thm:feature} provides an upper bound $\mathcal{O}(T^{(d+1)/(d+2)})$. This is tight when compared with the bound for continuum-armed bandit in $d$-dimension $\Theta(T^{(d+1)/(d+2)})$, which is a special case when there is only one content creator and one user~\citep{kleinberg2003value, kleinberg2013bandits}. 
\end{remark}
\begin{remark}
    Compared with the regret $\Theta(T^{2/3})$ in Theorem~\ref{thm:return}, we can see that the optimal feature-based contract may require more samples to approximate when the function class has large complexity.  This is consistent with the observation in the traditional contract theory, where linear contract is easier to learn than general contract~\citep{zhu2022sample}. Although the comparison between regret is clear, the optimal utilities for return-based contracts and feature-based contracts are not comparable due to different content generation functions. In Section~\ref{sec:full}, we provide a specific setting that makes the comparison possible.
\end{remark}
\section{Alternative Formulation without Smoothness}\label{sec:full}

In the above analysis, we rely on smoothness assumptions on the generated content. In the traditional contract design literature, such smoothness assumptions are usually not required. Instead, one assumes that the content created is always the content with the largest expected utility for the content creator or a sample from the distribution that maximizes the expected utility. 

A difficulty in adopting such a formulation in our setting is that the  expected utility for the content creator also depends on the content generated by the other content creators due to the flexibility in the recommendation policy. To address this challenge, we initiate the study with the case where all of the content is recommended, i.e., $K=S$. 
In this case, the recommendation policy satisfies $\pi(j, k, \{c_{k'}\}_{k'=1}^K) = 1$ for any $j\in[M], k\in[K], c_{k'}\in\mathcal{C}$. Thus the problem reduces to a contract design problem. We let $l_k(c)$ be the cost for the $k$-th content creator in creating content $c$.  The utility for the return-based contract is reduced to a linear contract:
\begin{align}
    u_{\mathsf{r}}(\alpha) &= (1-\alpha) \mathbb{E}\left[\sum_{k=1}^K  \sum_{j=1}^M R(j, c^{\star}_k)\right], \nonumber \\
    \text{ where } c^{\star}_k& = \argmax_{c\in\mathcal{C}} \alpha \cdot  \mathbb{E}\left[\sum_{j=1}^M R(j, c)\right]- l_k(c). \label{eq:return-full} 
\end{align}

The utility of  the feature-based contract can be written as
\begin{align}
u_{\mathsf{f}}(g) &  = \mathbb{E}\left[\sum_{k=1}^K \left(\left(\sum_{j=1}^M R(j, c_k^\star)\right)-g(c_k^\star)\right)\right], \nonumber \\
\text{ where } c^{\star}_k& = \argmax_{c\in\mathcal{C}} g(c) - l_k(c). \label{eq:feature-full}
\end{align}
We are able to compare the optimal utility for both cases when $\mathcal{G}$ contains the family of reward functions:
\begin{proposition}\label{prop:comp}
   Assume that the family of reward function is a subset of $\mathcal{G}$, i.e., $\{r_\alpha(c) = \alpha\cdot \mathbb{E}[\sum_{j=1}^M R(j, c)] \mid \alpha\in[0, 1]\} \subset \mathcal{G}$. Then we have
   \begin{align*}
     \max_{g\in\mathcal{G}} u_{\mathsf{f}}(g) \geq \max_{\alpha\in[0, 1]}   u_{\mathsf{r}}(\alpha).
   \end{align*}
\end{proposition}
The proof is deferred to Appendix~\ref{proof:comp}. 
It shows that a properly designed feature-based contract can lead to higher optimal utility than the return-based contract. However, the optimal feature-based contract usually requires more samples to approximate  compared to the return-based contract.
We make this point clear in the following theorem. The proof is deferred to Appendix~\ref{proof:full}.
\begin{theorem}\label{thm:full}
   Tthe regret for return-based contract in Equation (\ref{eq:return-full}) is $\Theta(T^{2/3})$. 
   
   For feature-based contract  in Equation (\ref{eq:feature-full}) with linear family $\mathcal{C} = \mathbb{B}^d$,  $\mathbb{E}[R(j, c)] = u_j^\top c$, $g\in\{g_\theta(c) = \theta^\top c \mid \theta \in\mathbb{B}^d\}$,  there exists some algorithm such that the regret is $\mathcal{O}(T^{2d/(2d+1)})$. 
\end{theorem}
The proof is based on a reduction to the traditional contract design problem in~\citet{zhu2022sample}. Following the same construction as~\citet[Theorem 5]{zhu2022sample},  we can also see that it requires at least $\Omega((1/\epsilon)^d)$ samples to achieve regret smaller than $\epsilon T$ for the feature-based contract in the worst case.    Combined with Proposition~\ref{prop:comp}, one can see that in this formulation with full recommendation, the optimal return-based contract requires less sample to approximate than the optimal feature-based contract. However, the  utility achieved by the optimal return-based contract is always smaller than that achieved by the optimal feature-based contract from a well-designed feature-based contract. This suggests an interesting trade-off: when the data size is small,  one may start with a return-based contract for fast convergence. As the data size grows larger, it is possible to further improve the utility by switching to a well-designed feature-based contract family. It is an interesting open problem how one can  design a contract family that combines the return-based contract and feature-based contract to achieve  the optimal trade-off in the online learning setting. 

 \section{Discussions and Conclusion}
In this paper, we initiate the study of online learning algorithms in a creator economy. We decompose the role of the platform into two parts: providing reasonable contracts to the content creator, and making appropriate recommendations to users. We analyze and provide tight regret guarantees for two types of contracts---return-based contracts and feature-based contracts. 

We briefly discuss several potential future directions in this line of research.

 \begin{itemize}
 \item For our alternative formulation without smoothness in Section~\ref{sec:full}, we  focus on the case of full recommendation. For the partial recommendation setting where the recommendation policy needs to be jointly optimized, the choice of the content generation function is less clear since the utility for each content creator depends on both the proposed contract and the recommendation policy, which further involve the content created by others. It is an interesting challenge to formulate and analyze this setting.
 \item 
Our study mostly focuses on the tabular case when the reward and the content generation for each user and content are independent. It would be interesting to study  the regret for the function approximation case,  especially on how the joint optimization of contract and recommendation can be combined with end-to-end neural network training.

\item  We assume that the content generation function is a deterministic function of the contract. This greatly simplifies the estimation of the content generation function. In practice, the content generation functions may be  random as well. This calls for new techniques to deal with the maximum in the recommendation policy.
     \item We assume that the content creators and users have a static content generation function and preferences throughout the process. In practice, the content creators may learn the users' preferences from the observations via some no-regret algorithms. Moreover, the users' preferences might be slowly changing as well. It is important to consider a full game-theoretic setting where the platform, users, and content creators are all strategic and maximizing their own utility.
 \end{itemize}
\newpage
\bibliography{ref}

\newpage 
\appendix
\section{Proof of Theorem~\ref{thm:return}}\label{proof:return}
\begin{proof}
\textbf{Upper Bound.}

We start with the upper bound for a return-based contract. 
By Hoeffding's inequality, for any $j\in[M]$, $k\in[K]$, $\alpha\in\mathcal{S}_\epsilon$, with probability at least $1-\delta$, 
\begin{align*}
 \hat r^{(t)}(j, c_k(\alpha))- \mathbb{E}[R(j, c_k(\alpha))]  \in\left[0, 2\sqrt{\frac{ 2\log(1/\delta)}{\sum_{s=1}^{t} \mathbbm{1}(\pi_s(j, k)=1,  \alpha_s = \alpha)}}\right].
\end{align*}
By taking a union bound over all $j\in[M]$, $k\in[K]$, $\alpha\in\mathcal{S}_\epsilon$, we know that the following event holds with probability at least $1-\delta$:
\begin{align*}
    E = \left\{\forall j\in[M], k\in[K], t\in[T], \alpha\in \mathcal{S}_\epsilon,  \hat r^{(t)}(j, c_k(\alpha))- \mathbb{E}[R(j, c_k(\alpha))]  \in\left[0, 2\sqrt{\frac{ 2\log(KMT/(\epsilon\delta))}{\sum_{s=1}^{t} \mathbbm{1}(\pi_s(j, k)=1,  \alpha_s = \alpha)}} \right]\right\}.
\end{align*}
Throughout the remaining proof, we condition on the event $E$. This indicates that $u_\mathsf{r}(\alpha,\pi)\leq \hat u_\mathsf{r}^{(t)}(\alpha,\pi)$ for any $t, \alpha$ and $\pi$. 
Let $\alpha_\epsilon^\star\in\mathcal{S}_\epsilon$ be the contract that is closest  in absolute value distance to $\alpha^\star$. From the definition of $\mathcal{S}_\epsilon$, we know that  
\begin{align*}
   | \alpha_\epsilon^\star -\alpha | \leq \epsilon.
\end{align*}
Let $\pi^\star_\epsilon = \argmax_{\pi\in\Pi} u_\mathsf{r}(\alpha^\star_\epsilon, \pi)$. We know that
\begin{align}
   & u_{\mathsf{r}}(\alpha^\star, \pi^\star) - u_{\mathsf{r}}(\alpha^\star_\epsilon, \pi^\star_\epsilon)  \nonumber \\ 
     = &(1-\alpha^\star) \mathbb{E}\left[\sum_{k=1}^K  \sum_{j=1}^M R(j, c_k(\alpha^\star))\cdot\pi^\star(j, k, \{c_{k'}(\alpha^\star)\}_{k'=1}^K)\right] \nonumber  \\ 
    & - (1-\alpha^\star_\epsilon) \mathbb{E}\left[\sum_{k=1}^K  \sum_{j=1}^M R(j, c_k(\alpha^\star_\epsilon))\cdot\pi^\star_\epsilon(j, k, \{c_{k'}(\alpha^\star_\epsilon)\}_{k'=1}^K)\right]  \nonumber \\ 
     \leq  &(1-\alpha^\star) \mathbb{E}\left[\sum_{k=1}^K  \sum_{j=1}^M R(j, c_k(\alpha^\star))\cdot\pi^\star(j, k, \{c_{k'}(\alpha^\star)\}_{k'=1}^K)\right] \nonumber  \\ 
    & - (1-\alpha^\star_\epsilon) \mathbb{E}\left[\sum_{k=1}^K  \sum_{j=1}^M R(j, c_k(\alpha^\star_\epsilon))\cdot\pi^\star(j, k, \{c_{k'}(\alpha^\star)\}_{k'=1}^K)\right]  \nonumber \\
   \leq  & KM \epsilon + \sum_{k=1}^K  \sum_{j=1}^M |\mathbb{E}[R(j, c_k(\alpha^\star))] - \mathbb{E}[R(j, c_k(\alpha^\star_\epsilon))] | \nonumber \\
     \leq & KM(L+1)\epsilon. \label{eq:alpha_discretization}
\end{align}
The regret can be decomposed as follows.
\begin{align*}
   &   \mathsf{Regret}_{\mathsf{r}}(T) \\
   \leq   &  \frac{KM}{\epsilon} + \sum_{t={1}/{\epsilon}}^{T} u_{\mathsf{r}}(\alpha^\star, \pi^\star) - u_{\mathsf{r}}(\alpha^\star_\epsilon, \pi^\star_\epsilon)  +  \sum_{t={1}/{\epsilon}}^{T} u_{\mathsf{r}}(\alpha^\star_\epsilon, \pi^\star_\epsilon)  - u_{\mathsf{r}}(\alpha_t, \pi_t)     \\ 
   =  & \frac{KM}{\epsilon} + KM(L+1)\epsilon T +  \sum_{t={1}/{\epsilon}}^{T} (u_{\mathsf{r}}(\alpha^\star_\epsilon, \pi^\star_\epsilon) - \hat u_{\mathsf{r}}^{(t)}(\alpha^\star_\epsilon, \pi^\star_\epsilon) ) \\ 
   & \quad + ( \hat u^{(t)}_{\mathsf{r}}(\alpha^\star_\epsilon, \pi^\star_\epsilon) - \hat u^{(t)}_{\mathsf{r}}(\alpha_t, \pi_t)) + (\hat u^{(t)}_{\mathsf{r}}( \alpha_t, \pi_t) -   u_{\mathsf{r}}(\alpha_t, \pi_t)) \\ 
   \leq & \frac{KM}{\epsilon} + KM(L+1)\epsilon T +  \sum_{t={1}/{\epsilon}}^{T}  (\hat u^{(t)}_{\mathsf{r}}( \alpha_t, \pi_t) -   u_{\mathsf{r}}(\alpha_t, \pi_t)).
\end{align*}
Here the last inequality uses the fact that $u_{\mathsf{r}}(\alpha^\star_\epsilon, \pi^\star_\epsilon) \leq \hat u_{\mathsf{r}}^{(t)}(\alpha^\star_\epsilon, \pi^\star_\epsilon) $  from event $E$. And $\hat u^{(t)}_{\mathsf{r}}(\alpha^\star_\epsilon, \pi^\star_\epsilon) \leq \hat u^{(t)}_{\mathsf{r}}(\alpha_t, \pi_t)$ comes from the fact that $\alpha_t, \pi_t$ is the maximizer. 

Now we proceed to the final sum. We have
\begin{align*}
    & \sum_{t=1}^{T}  (\hat u^{(t)}_{\mathsf{r}}( \alpha_t, \pi_t) -   u_{\mathsf{r}}(\alpha_t, \pi_t)) \\ 
    \leq &  \sum_{t=1}^T (1-\alpha_t)\Bigg( \sum_{k=1}^K  \sum_{j=1}^M \hat r^{(t)}(j, c_k(\alpha_t))\cdot\pi_t(j, k, \{c_{k'}(\alpha_t)\}_{k'=1}^K) \\
    & \quad - \mathbb{E}\left[\sum_{k=1}^K  \sum_{j=1}^M R(j, c_k(\alpha_t))\cdot\pi_t(j, k, \{c_{k'}(\alpha_t)\}_{k'=1}^K)\right] \Bigg)\\
  = &  \sum_{t=1}^T (1-\alpha_t)\left( \sum_{k=1}^K  \sum_{j=1}^M (\hat r^{(t)}(j, c_k(\alpha_t)) - R(j, c_k(\alpha_t)))\cdot\pi_t(j, k, \{c_{k'}(\alpha_t)\}_{k'=1}^K) \right)\\
  \leq  & 2\sum_{t=1}^T  \sum_{k=1}^K  \sum_{j=1}^M \sum_{\alpha\in\mathcal{S}_\epsilon}\sqrt{\frac{ 2\log(KMT/(\epsilon\delta))}{\sum_{s=1}^{t-1} \mathbbm{1}(\pi_s(j, k)=1,  \alpha_s = \alpha)}} \cdot\pi_t(j, k, \{c_{k'}(\alpha_t)\}_{k'=1}^K)\cdot \mathbbm{1}(\alpha_t = \alpha)  \\ 
   = & 2\sum_{t=1}^T  \sum_{k=1}^K  \sum_{j=1}^M \sum_{\alpha\in\mathcal{S}_\epsilon}\mathbbm{1}(\pi_t(j, k)=1,  \alpha_t = \alpha)\sqrt{\frac{ 2\log(KMT/(\epsilon\delta))}{\sum_{s=1}^{t-1} \mathbbm{1}(\pi_s(j, k)=1,  \alpha_s = \alpha)}} \\ 
 \leq &  C    \sqrt{ { KMS T\log(KMT/(\epsilon\delta))/\epsilon} }.
\end{align*}
Here the last inequality uses the fact that for any $k, j, \alpha$, the summation always takes the form of $\sum_{l=1}^{T(k, j, \alpha)}1/\sqrt{l}$. 
And thus the regret is upper bounded as follows.
\begin{align*}
    \mathsf{Regret}_{\mathsf{r}}(T) 
   \leq C({KM}/{\epsilon} + KM(L+1)\epsilon T + \sqrt{ { 2KMST\log(KMT/(\epsilon\delta))/\epsilon} }).  
\end{align*}
By taking  $\epsilon = T^{-1/3}$, we get that with probability at least $1-\delta$, 
\begin{align*}
     \mathsf{Regret}_{\mathsf{r}}(T)\leq  C( KM(L+1) + \sqrt{ { 2KMS\log(KMT/\delta)} })\cdot T^{2/3}.
\end{align*}

\textbf{Lower Bound.}

Now we present the lower bound, which is based on a direct reduction to the continuum-armed bandit problem.  Consider the case when there is only  $K=1$ content creator and $M=1$ user. In each round, the one content generated is always recommended to the user. We let $c(\alpha)=\alpha$, and 
\begin{align*}
    r(c) = \begin{cases}
        2\tilde r(1/2), & c \in [0, 1/2), \\
        \frac{\tilde r(c)}{1-c}, &  c\in[1/2, 1].
    \end{cases}
\end{align*}
In this case, we can see that $r(c)$ is $3L$-Lipschitz if $\tilde r(c)$ is $L$-Lipschitz and bounded in $[0,1]$. 
we can see that $$u(\alpha) =\begin{cases}
        \tilde r(1/2), & \alpha \in [0, 1/2), \\
       {\tilde r(\alpha)}, &  \alpha\in[1/2, 1].
    \end{cases} $$
This reduces to a continuum-armed bandit problem in the regime of $[1/2, 1]$, where the regret is lower bounded by $\Omega(T^{2/3})$~\citep{kleinberg2004nearly, kleinberg2013bandits}. 
\end{proof}

\section{Proof of Theorem~\ref{thm:feature}}\label{proof:feature}
\begin{proof}
\textbf{Upper bound.}

By Hoeffding's inequality, for any $j\in[M]$, $k\in[K]$, $g\in\mathcal{G}_\epsilon$, with probability at least $1-\delta$, 
\begin{align*}
 \hat r^{(t)}(j, c_k(g))- \mathbb{E}[R(j, c_k(g))]  \in\left[0, 2\sqrt{\frac{ 2\log(1/\delta)}{\sum_{s=1}^{t} \mathbbm{1}(\pi_s(j, k)=1,  g_s = g)}}\right].
\end{align*}
By taking union bound over all $j\in[M]$, $k\in[K]$, $g\in\mathcal{G}_\epsilon$, we know that the following event holds with probability at least $1-\delta$:
\begin{align*}
    E = \left\{\forall j\in[M], k\in[K], t\in[T], g\in \mathcal{G}_\epsilon,  \hat r^{(t)}(j, c_k(g))- \mathbb{E}[R(j, c_k(g))]  \in\left[0, 2\sqrt{\frac{ 2d(\mathcal{G})\log(KMT/(\epsilon\delta))}{\sum_{s=1}^{t} \mathbbm{1}(\pi_s(j, k)=1,  g_s = g)}} \right]\right\}.
\end{align*}
Throughout the remaining proof, we condition on the event $E$. This indicates that $u_\mathsf{r}(g,\pi)\leq \hat u_\mathsf{f}^{(t)}(g,\pi)$ for any $t, g$ and $\pi$. 

Let $g_\epsilon^\star\in\mathcal{G}_\epsilon$ be the contract that is closest  in infinity norm to $g^\star$. From the definition of $\mathcal{G}_\epsilon$, we know that  
\begin{align*}
   \| g_\epsilon^\star -g \|_\infty \leq \epsilon.
\end{align*}
Let $\pi^\star_\epsilon = \argmax_{\pi\in\Pi} u_\mathsf{r}(g^\star_\epsilon, \pi)$. We know that
\begin{align}
   & u_{\mathsf{f}}(g^\star, \pi^\star) - u_{\mathsf{f}}(g^\star_\epsilon, \pi^\star_\epsilon)  \nonumber \\ 
     = & \mathbb{E}\left[\sum_{k=1}^K  \sum_{j=1}^M R(j, c_k(g^\star))\cdot\pi^\star(j, k, \{c_{k'}(g^\star)\}_{k'=1}^K)- g^\star(c_k(g^\star))\right]  \nonumber  \\ 
    & -  \mathbb{E}\left[\sum_{k=1}^K  \sum_{j=1}^M R(j, c_k(g^\star_\epsilon))\cdot\pi^\star_\epsilon(j, k, \{c_{k'}(g^\star_\epsilon)\}_{k'=1}^K) - g^\star_\epsilon(c_k(g^\star_\epsilon))\right]  \nonumber \\ 
     \leq  & \mathbb{E}\left[\sum_{k=1}^K  \sum_{j=1}^M R(j, c_k(g^\star))\cdot\pi^\star(j, k, \{c_{k'}(g^\star)\}_{k'=1}^K)\right]  -\mathbb{E}\left[\sum_{k=1}^K  \sum_{j=1}^M R(j, c_k(g^\star_\epsilon))\cdot\pi^\star(j, k, \{c_{k'}(g^\star)\}_{k'=1}^K)\right] \nonumber  \\ 
    & + \sum_{k=1}^K|g^\star(c_k(g^\star)) -  g^\star_\epsilon(c_k(g^\star_\epsilon))|  \nonumber \\
   \leq  &   \sum_{k=1}^K  \sum_{j=1}^M |\mathbb{E}[R(j, c_k(g^\star)) - R(j, c_k(g^\star_\epsilon)) ]| + \sum_{k=1}^K|g^\star_\epsilon(c_k(g^\star)) -  g^\star_\epsilon(c_k(g^\star_\epsilon))|  + L_2\|c_k(g^\star) - c_k(g^\star_\epsilon)\|_2\nonumber \\
     \leq & KML_1L_3\epsilon + \sum_{k=1}^K 2L_2 L_3\epsilon \nonumber \\ 
     \leq & K(ML_1L_3 + 2L_2L_3)\epsilon.
     \label{eq:g_discretization}
\end{align}

The regret can be decomposed as follows.
\begin{align*}
   &   \mathsf{Regret}_{\mathsf{f}}(T) \\
   \leq   &  KM|\mathcal{G}_\epsilon| + \sum_{t=|\mathcal{G}_\epsilon|+1}^{T} u_{\mathsf{f}}(g^\star, \pi^\star) - u_{\mathsf{f}}(g^\star_\epsilon, \pi^\star_\epsilon)  +  \sum_{t=|\mathcal{G}_\epsilon|+1}^{T} u_{\mathsf{f}}(g^\star_\epsilon, \pi^\star_\epsilon)  - u_{\mathsf{f}}(g_t, \pi_t)     \\ 
   =  & KM|\mathcal{G}_\epsilon| +  K(ML_1L_3 + 2L_2L_3)\epsilon T +  \sum_{t=|\mathcal{G}_\epsilon|+1}^{T}( (u_{\mathsf{f}}(g^\star_\epsilon, \pi^\star_\epsilon) - \hat u_{\mathsf{f}}^{(t)}(g^\star_\epsilon, \pi^\star_\epsilon) ) \\
   & \qquad + ( \hat u^{(t)}_{\mathsf{f}}(g^\star_\epsilon, \pi^\star_\epsilon) - \hat u^{(t)}_{\mathsf{f}}(g_t, \pi_t)) + (\hat u^{(t)}_{\mathsf{f}}( g_t, \pi_t) -   u_{\mathsf{f}}(g_t, \pi_t))) \\ 
   \leq & KM|\mathcal{G}_\epsilon| +  K(ML_1L_3 + 2L_2L_3)\epsilon T +  \sum_{t=|\mathcal{G}_\epsilon|+1}^{T}  (\hat u^{(t)}_{\mathsf{f}}( g_t, \pi_t) -   u_{\mathsf{f}}(g_t, \pi_t)).
\end{align*}
Here the last inequality uses the fact that $u_{\mathsf{f}}(g^\star_\epsilon, \pi^\star_\epsilon) \leq \hat u_{\mathsf{f}}^{(t)}(g^\star_\epsilon, \pi^\star_\epsilon) $  from event $E$. And $\hat u^{(t)}_{\mathsf{f}}(g^\star_\epsilon, \pi^\star_\epsilon) \leq \hat u^{(t)}_{\mathsf{f}}(g_t, \pi_t)$ comes from the fact that $g_t, \pi_t$ is the maximizer. 

Now we proceed to the final sum. We have
\begin{align*}
    & \sum_{t=1}^{T}  (\hat u^{(t)}_{\mathsf{f}}( g_t, \pi_t) -   u_{\mathsf{f}}(g_t, \pi_t)) \\ 
    \leq &  \sum_{t=1}^T \left( \sum_{k=1}^K  \sum_{j=1}^M \hat r^{(t)}(j, c_k(g_t))\cdot\pi_t(j, k, \{c_{k'}(g_t)\}_{k'=1}^K)- \mathbb{E}\left[\sum_{k=1}^K  \sum_{j=1}^M R(j, c_k(g_t))\cdot\pi_t(j, k, \{c_{k'}(g_t)\}_{k'=1}^K)\right] \right)\\
  = &  2\sum_{t=1}^T\left( \sum_{k=1}^K  \sum_{j=1}^M (\hat r^{(t)}(j, c_k(g_t)) - \mathbb{E}[R(j, c_k(g_t))])\cdot\pi_t(j, k, \{c_{k'}(g_t)\}_{k'=1}^K) \right)\\
  \leq  & 2\sum_{t=1}^T  \sum_{k=1}^K  \sum_{j=1}^M \sum_{g\in\mathcal{G}_\epsilon}\sqrt{\frac{ 2d(\mathcal{G})\log(KMT/(\epsilon\delta))}{\sum_{s=1}^{t-1} \mathbbm{1}(\pi_s(j, k)=1,  g_s = g)}} \cdot\pi_t(j, k, \{c_{k'}(g_t)\}_{k'=1}^K)\cdot \mathbbm{1}(g_t = g)  \\ 
   = & 2\sum_{t=1}^T  \sum_{k=1}^K  \sum_{j=1}^M \sum_{g\in\mathcal{G}_\epsilon}\mathbbm{1}(\pi_t(j, k)=1,  g_t = g)\sqrt{\frac{ 2d(\mathcal{G})\log(KMT/(\epsilon\delta))}{\sum_{s=1}^{t-1} \mathbbm{1}(\pi_s(j, k)=1,  g_s = g)}} \\ 
 \leq &  C \sqrt{ { KM|\mathcal{G}_\epsilon|Td(\mathcal{G})\log(KMT/(\epsilon\delta))} }.
\end{align*}
Here the last inequality uses the fact that for any $k, j, g$, the summation always takes the form of $\sum_{l=1}^{T(k, j, g)}1/\sqrt{l}$. 
And thus the regret is upper-bounded as follows.
\begin{align*}
    \mathsf{Regret}_{\mathsf{f}}(T) 
   & \leq C\cdot (KM|\mathcal{G}_\epsilon| +  K(ML_1L_3 + 2L_2L_3)\epsilon T + \sqrt{ { KM|\mathcal{G}_\epsilon|Td(\mathcal{G})\log(KMT/(\epsilon\delta))} }) \\ 
   & \leq  C\cdot \left(KM/{\epsilon}^{d(\mathcal{G})} +  K(ML_1L_3 + 2L_2L_3)\epsilon T + \sqrt{ { KMTd(\mathcal{G})\log(KMT/(\epsilon\delta)) / \epsilon^{d(\mathcal{G})}} }\right).
\end{align*}
By taking $\epsilon = T^{-1/(d(\mathcal{G})+2)}$, we get that with probability at least $1-\delta$, 
\begin{align*}
     \mathsf{Regret}_{\mathsf{f}}(T)\leq C\cdot (KML_1L_3 + KL_2L_3 +  \sqrt{ KMd(\mathcal{G})\log(KMT/\delta)})\cdot T^{(d(\mathcal{G})+1)/(d(\mathcal{G})+2)}.
\end{align*}

\end{proof}

\section{Proof of Proposition~\ref{prop:comp}}\label{proof:comp}

The proof is straightforward from the definitions. Let $\mathcal{R} = \{r_\alpha(c) = \alpha\cdot \mathbb{E}[\sum_{j=1}^M R(j, c)] \mid \alpha\in[0, 1]\}$. From the assumption that $\mathcal{R}\subset \mathcal{G}$,  we have
\begin{align*}
\mathsf{max}_{g\in\mathcal{G}} u_{\mathsf{f}}(g) \geq & \mathsf{max}_{g\in\mathcal{R}} u_{\mathsf{f}}(g) = \max_\alpha (1-\alpha) \mathbb{E}\left[\sum_{k=1}^K  \sum_{j=1}^M R(j, c^{\star}_k)\right], \nonumber \\
    \text{ where } c^{\star}_k& = \argmax_{c\in\mathcal{C}} \alpha \cdot  \mathbb{E}\left[\sum_{j=1}^M R(j, c)\right]- l_k(c).
\end{align*}

This finishes the proof by noting that the last equality is exactly $\max_{\alpha} u_{\mathsf{r}}(\alpha)$.

\section{Proof of Theorem~\ref{thm:full}}\label{proof:full}
\begin{proof}
\textbf{Return-based Contract.}

Let \begin{align}
   u_{\mathsf{r}, k} (\alpha) &= (1-\alpha) \mathbb{E}\left[  \sum_{j=1}^M R(j, c^{\star}_k)\right], \nonumber \\
    \text{ where } c^{\star}_k& = \argmax_{c\in\mathcal{C}} \alpha \cdot  \mathbb{E}\left[\sum_{j=1}^M R(j, c)\right]- l_k(c). 
\end{align}   
Then we know that $  u_{\mathsf{r}}(\alpha) = \sum_{k=1}^K     u_{\mathsf{r}, k} (\alpha).$
We prove the following lemma first.
\begin{lemma}
    For any $0\leq \alpha < \alpha' \leq 1$, we have
    \begin{align*}
        \frac{u_{\mathsf{r}, k}(\alpha)}{1-\alpha} \leq  \frac{u_{\mathsf{r}, k}(\alpha')}{1-\alpha'}.
    \end{align*}
\end{lemma}
\begin{proof}
    From the optimality of $c_k^\star(\alpha)$, we know that
    \begin{align*}
        \alpha \cdot  \mathbb{E}\left[\sum_{j=1}^M R(j, c_k^\star(\alpha))\right]- l_k(c_k^\star(\alpha)) \geq   \alpha \cdot  \mathbb{E}\left[\sum_{j=1}^M R(j, c_k^\star(\alpha'))\right]- l_k(c_k^\star(\alpha')).
    \end{align*}
On the other hand, from the optimality of of $c_k^\star(\alpha')$, we know that
    \begin{align*}
        \alpha' \cdot  \mathbb{E}\left[\sum_{j=1}^M R(j, c_k^\star(\alpha'))\right]- l_k(c_k^\star(\alpha')) \geq   \alpha' \cdot  \mathbb{E}\left[\sum_{j=1}^M R(j, c_k^\star(\alpha))\right]- l_k(c_k^\star(\alpha)).
    \end{align*}
Now by summing up the above two equations, we arrive at
    \begin{align*}
        (\alpha'-\alpha) \cdot  \left(\mathbb{E}\left[\sum_{j=1}^M R(j, c_k^\star(\alpha'))\right]-  \mathbb{E}\left[\sum_{j=1}^M R(j, c_k^\star(\alpha))\right]\right) >0.
    \end{align*}
This implies that $\mathbb{E}\left[\sum_{j=1}^M R(j, c_k^\star(\alpha'))\right] >  \mathbb{E}\left[\sum_{j=1}^M R(j, c_k^\star(\alpha))\right]$, which is equivalent to  $\frac{u_{\mathsf{r}, k}(\alpha)}{1-\alpha} <  \frac{u_{\mathsf{r}, k}(\alpha')}{1-\alpha'}$.
\end{proof}
The above lemma implies that $u_{\mathsf{r}}$ also satisfies $\frac{u_{\mathsf{r}}(\alpha)}{1-\alpha} <  \frac{u_{\mathsf{r}}(\alpha')}{1-\alpha'}$. The rest of the proof follows the same upper bound argument as~\citet[Theorem 6]{zhu2022sample}, which relies on the right-continuity to design an algorithm based on the uniform discretization. 

\textbf{Feature-based Contract.}
We abuse the notation a bit and let  $u_{\mathsf{f}}(\theta) = u_{\mathsf{f}}(g_\theta) $. Let  
\begin{align*}
u_{\mathsf{f}, k}(\theta) &  =  \left(\sum_{j=1}^M u_j - \theta\right)^\top  c_k^\star,  \\
\text{ where } c^{\star}_k& = \argmax_{c\in\mathcal{C}} \theta^\top c - l_k(c).  
\end{align*}

We prove the following lemma:
\begin{lemma}\label{lem:gen_inequality}
For any two contracts $\theta, \theta +\gamma$ and any fixed $i$, one has 
\begin{align*}
      (c_k^\star(\theta+\gamma) - c_k^\star(\theta) )^\top \cdot \gamma \geq  0.
\end{align*}
\end{lemma}
\begin{proof}
Recall that $c^\star_k(\theta)$ is the optimal action under contract $\theta$. Thus one has 
\begin{align*}
    \theta^\top c^\star_k(\theta) - l_k( c^\star_k(\theta)) \geq
    \theta^\top c^\star_k(\theta+\gamma) - l_k( c^\star_k(\theta+\gamma)) .
\end{align*} 
Similarly, from that $c^\star_k(\theta+\gamma)$ is the optimal action under contract $\theta+\gamma$, we have
\begin{align*}
    (\theta+\gamma)^\top c^\star_k(\theta+\gamma) - l_k( c^\star_k(\theta+\gamma)) \geq
    (\theta+\gamma)^\top c^\star_k(\theta) - l_k( c^\star_k(\theta)) .
\end{align*} 
Summing up the above two equations, we arrive at the desired result. 
\end{proof}

With the help of Lemma~\ref{lem:gen_inequality}, we can show that under some specific direction, the utility is approximately continuous.
\begin{lemma}[Continuity of the utility function]\label{lem:gen_continuity}
Let $\gamma = \alpha(\sum_{j=1}^M u_j-\theta) +  \eta$ for some $\alpha\in(0, 1]$. For any $\theta\in\mathbb{B}^d$, we have 
\begin{align*}
    u_{\mathsf{f}, k}(\theta) - u_{\mathsf{f}, k}(\theta+\gamma) \leq 2\left( \|\vec\gamma\|_2 + \frac{\|\eta\|_2}{\alpha}\right).
\end{align*}
\end{lemma}
\begin{proof}
From Lemma~\ref{lem:gen_inequality}, we know that for $\gamma = \alpha(\sum_{j=1}^M u_j-\theta) +  \eta$, we have 
\begin{align*}
      (c_k^\star(\theta+\gamma) - c_k^\star(\theta) )^\top   \gamma \geq  0.
\end{align*}
Rearranging the above formula gives us  that 
\begin{align*}
     (c_k^\star(\theta) - c_k^\star(\theta+\gamma) )^\top \left(\sum_{j=1}^M u_j-\theta\right) \leq   (c_k^\star(\theta+\gamma) - c_k^\star(\theta))^\top \eta / \alpha  \leq 2\|\eta\|_2 / \alpha.
\end{align*}
Thus we have
\begin{align}
     u_{\mathsf{f}, k}(\theta) - u_{\mathsf{f}, k}(\theta+\gamma) & = \left(\sum_{j=1}^M u_j - \theta\right)^\top  c_k^\star(\theta) - \left(\sum_{j=1}^M u_j - \theta-\gamma\right)^\top  c_k^\star(\theta+\gamma)\\
     & \leq \left(\sum_{j=1}^M u_j - \theta\right)^\top  c_k^\star(\theta) - \left(\sum_{j=1}^M u_j - \theta \right)^\top  c_k^\star(\theta+\gamma) + \|\gamma\|_2 \\
     & \leq 2\|\eta\|_2/\alpha + \|\gamma\|_2.
\end{align} 
This finishes the proof of the lemma.
\end{proof}
The rest of the proof follows directly~\citet[Theorem 4]{zhu2022sample}, which applies a discretization based on spherical code to achieve rate $\mathcal{O}(T^{2d/(2d+1)})$. 
\end{proof}

\end{document}